\documentclass[12pt,a4paper,reqno]{amsart}
\usepackage{amsfonts}
\usepackage{amsthm}
\usepackage{amsmath}
\usepackage{times}
\usepackage{amscd}
\usepackage[latin2]{inputenc}
\usepackage{t1enc}
\usepackage{enumitem}
\usepackage[mathscr]{eucal}
\usepackage{indentfirst}
\usepackage{graphicx}
\usepackage{graphics}
\usepackage{pict2e}
\usepackage{epic}
\numberwithin{equation}{section}
\usepackage[margin=2.9cm]{geometry}
\usepackage{epstopdf} 
\usepackage{verbatim}
\usepackage{cancel}
\usepackage{amssymb}
\usepackage{xcolor}
\usepackage{physics}
\usepackage{dsfont}
\usepackage{hyperref}

\setitemize{itemsep=5pt, topsep=7pt}

\theoremstyle{plain}
\newtheorem{theo}{Theorem}[section]
\newtheorem{lem}{Lemma}[section]

\newtheorem{prop}{Proposition}[section]

\theoremstyle{definition}

\newtheorem{rem}{Remark}[section]

\newcommand{\Lminus}{L_{-}}
\newcommand{\Lplus}{L_{+}}
\newcommand{\Lpplus}{\mathscr{L}_{+}}
\newcommand{\En}{\mathscr{E}_R}
\newcommand{\tEn}{\mathscr{\tilde E}_R}
\newcommand{\R}{\mathbb{R}}
\newcommand{\Rtre}{\mathbb{R}^{3}}
\newcommand{\Hspace}{H^1(\mathbb{R}^{3})}
\newcommand{\Hball}{H^1_0(B_R)}
\newcommand{\Lball}{L^2(B_R)}

\newcommand{\Om}{\Omega}

\newcommand{\Gker}{(-\Delta_{B_R})^{-1}(x,y)}

\newcommand{\minP}{\phi_R}
\newcommand{\HP}{H_R}
\newcommand{\QLplus}{Q\Lplus Q}
\newcommand{\restr}[1]{|_{#1}}
\newcommand{\XL}{X^{(l)}}

\newcommand{\tLplusL}{\tilde{L}_+^{(l)}}
\newcommand{\LplusL}{L_+^{(l)}}
\newcommand{\LminusL}{L_-^{(l)}}
\newcommand{\XLone}{X^{(l)}_1}
\newcommand{\XLtwo}{X^{(l)}_2}
\newcommand{\dxi}{\frac{\partial}{\partial x_i}}
\newcommand{\tLplus}{\tilde{L}_+}

\newcommand\id{\mathds{1}}
\DeclareMathOperator{\spn}{span}
\DeclareMathOperator{\infspec}{inf\,spec}
\DeclareMathOperator{\ran}{ran}
\DeclareMathOperator{\dom}{dom}

\begin{document}

\title[Minimizers of the Pekar Functional on a Ball]{Uniqueness and Non-degeneracy of Minimizers of the Pekar Functional on a Ball}

\author{Dario Feliciangeli and Robert Seiringer}
\address{IST Austria, Am Campus 1, 3400 Klosterneuburg, Austria}

\begin{abstract}
We consider the Pekar functional on a ball in $\R^3$. We prove uniqueness of minimizers, and a quadratic lower bound in terms of the distance to the minimizer. The latter follows from non-degeneracy of the Hessian at the minimum.
\end{abstract}

\date{April 18, 2019}

\maketitle

\section{Statement of the Problem and Main Results}

The Pekar functional arises as a classical approximation of the ground state energy of the Fr\"ohlich polaron model. Works of Donsker and Varadhan \cite{donsker1983asymptotics} and Lieb and Thomas \cite{lieb1997exact} show that this approximation is correct, up to lower order corrections, in the strong coupling limit. Motivated by \cite{frank2019quantum}, where quantum corrections to the classical approximation were studied in the case of a polaron confined to a bounded subset of $\R^3$, we consider here the Pekar functional on a ball. Our goal is to extend the results of \cite{lieb1977existence} and  \cite{lenzmann2009uniqueness}, where the problem is treated on $\Rtre$, to this case. In particular, we refer to the existence and uniqueness 
of minimizers 
(proved in \cite{lieb1977existence}) and to the coercivity around these minimizers (proved in \cite{lenzmann2009uniqueness}). 

Let $B_R$ denote the open ball of radius $R$ centered at the origin. 
We will consider Dirichlet boundary conditions on $B_R$, which corresponds to working with functions  $\phi \in  H^1_0(B_R)$. The Pekar  functional is 
\begin{equation}
\label{energy}
\En(\phi)=\int_{B_R}|\nabla \phi|^2 dx - 4\pi\int_{B_R} \int_{B_R} \Gker|\phi(x)|^2|\phi(y)|^2 dx dy\,,
\end{equation}
where $\Gker$ denotes the integral kernel of the inverse of the Dirichlet Laplacian on $B_R$. Explicitly, 
\begin{equation}\label{DGF}
\Gker= \frac 1{4\pi} \left( \frac {1}{|x-y|} - \frac{1}{\big| \frac{|y|}{R} x- \frac{R}{|y|} y\big|}\right)\,.
\end{equation}
Our main results are as follows:

\begin{theo}
\label{minimizer}
For any $R>0$, there exists a minimizer $0\leq\phi_R\in C^{\infty}(B_R)\cap H^1_0(B_R)$ such that
\begin{equation}
\En(\phi_R)= E_R := \inf\{\En(\phi): \phi\in H^1_0(B_R), \|\phi\|_2=1\}.
\end{equation}
Moreover, $\phi_R$ is the unique positive minimizer, it is strictly positive, radial and decreasing. Any other minimizer of $\En$ differs from $\phi_R$ by multiplication by a constant phase.
\end{theo}

\begin{theo}
\label{coercivitythm}
For any $R>0$, there exists a $K_R>0$ such that the  coercivity estimate  
\begin{equation}\label{eq:thm2}
\En(\phi)\geq \En(\phi_R)+K_R \min_{\theta\in [0,2\pi)}\int_{B_R} |\nabla (e^{i\theta}\phi_R-\phi)|^2 dx.
\end{equation}
holds for any $L^2$-normalized $\phi \in \Hball$.
\end{theo}

The study of this problem is motivated by the recent work \cite{frank2019quantum}, where lower order corrections to the ground state energy of the Fr\"ohlich polaron model in the strong coupling limit are investigated.  In particular, in \cite{frank2019quantum}, Theorem \ref{minimizer} and, in a slightly weaker form,  Theorem \ref{coercivitythm} are taken as  assumptions and are conjectured to hold for a large class of domains (e.g. convex domains). The goal of our work is to show that, at least in the case of balls, these assumptions hold true. 

\begin{rem}
Our results  apply equally if we consider instead of $\En$ the Pekar functional on the full space $\R^3$ {restricted} to $\Hball$. This amounts to considering, for $\phi \in \Hball$ with $\|\phi\|_2=1$, the functional 
\begin{equation}\label{tenergy}
 \tEn(\phi)= \int_{B_R} |\nabla \phi|^2 dx - \int_{B_R}\int_{B_R} \frac {|\phi(x)|^2|\phi(y)|^2}{|x-y|} dx dy\,.
\end{equation}
The necessary modifications in the proofs will be  explained in Remark~\ref{alsootherkernel}. 
\end{rem}

\begin{rem}
In the context of nonlinear Schr\"odinger equations with \emph{local} nonlinearities the non-degeneracy of linearizations is a well known fact (see \cite{weinstein1985modulational}, \cite{chang2007spectra}). Our model does not fall into this category since the linearization we have to deal with has a \emph{non-local} nature. Nevertheless, using similar techniques as the ones used in \cite{lenzmann2009uniqueness} and \cite{wei2009strongly}, the radial symmetry of the problem still allows to conclude non-degeneracy. 
\end{rem}

\section{Existence and Properties of Minimizers}

We start by showing that minimizers exist. This can be done with standard techniques; the proof is actually  easier on balls (because of compactness) than it is on the whole space.
It will be convenient to introduce the notation
\begin{equation}
T_R(\phi)=\int_{B_R}|\nabla\phi|^2 dx, \hspace{7mm} W_R(\phi)=4\pi\int_{B_R} \int_{B_R} \Gker|\phi(x)|^2|\phi(y)|^2 dx dy\,.
\end{equation}

\begin{theo}
\label{existencethm}
For any $R>0$, there exists an $L^2$-normalized $\phi_R\in \Hball$ such that $\En(\phi_R)=E_R$.
\end{theo}

\begin{proof}
Let $\phi\in \Hball$. By the Hardy--Littlewood--Sobolev, H\"older and Sobolev inequalities, 
\begin{equation}
\label{Wbound}
W_R(\phi)\leq\int_{B_R}\int_{B_R}\frac {|\phi(x)|^2|\phi(y)|^2}{|x-y|}dxdy\leq C\|\phi^2\|_{6/5}^2\leq C \|\phi\|_6\|\phi\|_2^3\leq C \|\nabla\phi\|_2\|\phi\|_2^3
\end{equation}
for suitable constants $C$ (which may take different values at different appearances). 
Hence 
\begin{equation}
\label{Energybound}
\En(\phi)=T_R(\phi)-W_R(\phi)\geq \tfrac 12 \|\nabla \phi\|_2^2-C\|\phi\|_2^6\,.
\end{equation}
We conclude that the functional is bounded from below for $L^2$-normalized functions, and that any minimizing sequence is bounded in $H^1_0(B_R)$. The Rellich--Kondrachov  and Ba\-nach--Alaoglu Theorems allow us to conclude that any minimizing sequence $\phi_{n}$ has a subsequence that converges to some $\phi_R$,  strongly in $L^p(B_R)$ for every $p\in[1,6)$  
and weakly in $\Hball$. 
Hence we have  $\|\phi_R\|_{2}=1$ and, by lower semicontinuity of the norm w.r.t. weak convergence, $T_R(\phi_R)\leq \liminf_{n\to \infty} T_R(\phi_{n})$. Moreover, with $\rho_{n}:=|\phi_{n}|^2$ and  $\rho:=|\phi_R|^2$ we have 
\begin{align}\nonumber
|W_R(\phi_{n})-W_R(\phi_R)|&=|\expval{-\Delta_{B_R}^{-1}}{\rho_{n}}-\expval{-\Delta_{B_R}^{-1}}{\rho}|=|\bra{\rho_n-\rho} \!-\!\Delta_{B_R}^{-1}\ket {\rho_n+\rho}|\\
&\leq C_R \|\rho_n-\rho\|_2\|\rho_n+\rho\|_2\to 0.
\end{align} 
Here, we used that $-\Delta_{B_R}^{-1}$ is a bounded operator (actually compact) on $L^2(B_R)$  and that $\rho_n\to \rho$ in $L^2$. Putting these pieces together, we  conclude that $\phi_R$ is a minimizer, since
\begin{equation}
\En(\phi_R)\leq \liminf_{n\to \infty}\En(\phi_{n})=E_R, \hspace{3mm} \phi_R\in H^1_0(B_R)\hspace{3mm} \text{and} \hspace{3mm} \|\phi_R\|_{L^2(B_R)}=1\,.
\end{equation}
\end{proof}

\begin{rem}\label{existencealldomains}
We point out  that this proof extends \emph{verbatim} to any bounded domain, the fact that we are working on $B_R$ does not play any role. This is  not true for the uniqueness statements that will come in the next sections, however. 
\end{rem}

Having established existence, we proceed to investigate  properties of  minimizers. 
\begin{lem}
	\label{regularity}
Let $\phi\in \Hball$, $\|\phi\|_{2}=1$ and $\En(\phi)=E_R$. Then $\phi$ satisfies the equation
\begin{equation}
\label{ELeq}
\left(-\Delta - e_{\phi} -2V_{\phi}\right)\phi=0
\end{equation}
on $B_R$, with 
\begin{equation}\label{def:ep}
e_{\phi}:=T_R(\phi)-2W_R(\phi)
\end{equation}
and
\begin{equation}\label{def:Vp}
V_{\phi}(x):=4\pi\int_{B_R} \Gker |\phi(y)|^2 dy\,.
\end{equation}
Moreover, $\phi \in C^{\infty}(B_R)$ and if $\phi\geq 0$ then $\phi>0$ on $B_R$.
\end{lem}

\begin{proof}
Eq.~\eqref{ELeq} is  the Euler--Lagrange equation associated to our minimization problem and its derivation is standard. Since $\phi\in \Hball$, $|\phi|^2$ is in $L^2$ (by Sobolev embeddings). Moreover, the function $y\mapsto  \Gker$ is bounded in $\Lball$ uniformly in $x$. We can  thus  conclude that $V_{\phi} \in L^{\infty}(B_R)$. Since $\phi$ solves \eqref{ELeq} and is in $\Hball$,  it satisfies
\begin{equation}\label{concl}
\phi(x)=\int_{B_R} (-\Delta_{B_R} + \lambda)^{-1}(\lambda+ e_{\phi}+2V_{\phi}(y))\phi(y) dy
\end{equation}
for any $\lambda> -\infspec(-\Delta_{B_R})$, 
and by bootstrapping we can conclude that $\phi\in C^{\infty}(B_R)$. 
Finally, suppose $\phi\geq 0$.
Choosing $\lambda> -e_\phi$ and 
exploiting the fact that  $(-\Delta_{B_R} + \lambda)^{-1} $ is  positivity improving, \eqref{concl} implies that $\phi>0$. 
\end{proof}

Next we shall exploit the radial symmetry of the problem. Similarly to  \cite{lieb1977existence}, we will make use of 
the tool of symmetric decreasing rearrangement \cite[Chapter 3]{lieb2001analysis}. For any measurable positive function $f$, we will denote its symmetric decreasing rearrangement as $f^*$. If $f$ is complex-valued, we will denote $f^*=|f|^*$. We recall the following Theorem, known as Talenti's Inequality  \cite{talenti1976elliptic}. In the strict form stated here, it is proved in \cite[Theorem 3]{alvino1986remark} (see also \cite{kesavan1988some} and \cite{kesavan1991comparison}). The result in \cite[Theorem 3]{alvino1986remark} is actually more general, but for simplicity we only state the version needed for our purposes.  

\begin{theo}[Talenti's Inequality]
\label{Talenti}
Let $0\leq f\in L^2(B_R)$, and let $u,v \in H^1_0(B_R)$ solve
\begin{equation}\label{eq:talenti}
\begin{cases}
-\Delta u=f \hspace{10pt} &x\in B_R,\\
u=0 \hspace{10pt} &x\in \partial B_R,
\end{cases}
\quad \text{and} \ 
\begin{cases}
-\Delta v=f^* \hspace{10pt} &x\in B_R,\\
v=0 \hspace{10pt} &x\in \partial B_R.
\end{cases}
\end{equation}
Then $u^*\leq v$ a.e. in $B_R$.
If additionally  $u^*(x_0)=v(x_0)$ for some $x_0$ with ${|x_0|=t\in(0,R)}$, then $u(x)=v(x)$ and $f(x)=f^*(x)$ for all $x$ with $t\leq|x|\leq R$.
\end{theo}

With these tools in hand, we can show the following key Proposition, which will be essential to prove  uniqueness of minimizers.  

\begin{prop}
\label{radiality}
Let $\phi\in \Hball$ be a minimizer of $\En$. Then $|\phi|=\phi^*$ and there exists a $\theta\in[0,2\pi)$ such that $\phi=e^{i\theta}|\phi|$.
\end{prop}

\begin{proof}
Clearly, for any $\psi\in\Hball$, $W_R(\psi)=W_R(|\psi|)$, and it is easy to see (\cite[Theorem 7.8]{lieb2001analysis}) that $T_R(\psi)\geq T_R(|\psi|)$. Hence, $\En(\psi)\geq \En(|\psi|)$. 
To proceed, we exploit the properties of symmetric decreasing rearrangements. The P\'olya--Szeg\H o  inequality \cite[Lem.~7.17]{lieb2001analysis}  states that 
\begin{equation}
\label{est1}
T_R(|\psi|)\geq T_R(\psi^*).
\end{equation} 
We claim that also
\begin{equation}
\label{est1b}
W_R(|\psi|)\leq W_R(\psi^*), 
\end{equation} 
with equality if and only if $|\psi| =\psi^*$. To see this we define 
\begin{equation}
u(x):=\int_{B_R} \Gker |\psi(y)|^2 dy \hspace{4mm} \text{and} \hspace{4mm} v(x):=\int_{B_R} \Gker \psi^*(y)^2 dy.
\end{equation}
These functions satisfy \eqref{eq:talenti} with $f(x)= |\psi(x)|^2$. 
By Theorem \ref{Talenti}, we  conclude that $u^*\leq v$. Applying first this estimate and then the Hardy--Littlewood rearrangement inequality \cite[Thm.~3.4]{lieb2001analysis}, we obtain
\begin{align}
\label{est2}
W_R(\psi^*)&=4\pi\int_{B_R} \psi^*(x)^2 v(x) dx \geq 4\pi\int_{B_R} \psi^*(x)^2 u^*(x) dx \nonumber\\
&\geq 4\pi \int_{B_R} |\psi(x)|^2 u(x) dx=W_R(|\psi|). 
\end{align}
To have equality in \eqref{est2}, we must have $v=u^*$ on the support  of $\psi^*$, which contains a non-empty ball centered at the origin. Hence the second part of Theorem~\ref{Talenti} implies that $v=u$ and thus $|\psi|=\psi^*$ on $B_R$, as claimed. For any $\psi \in \Hball$, we conclude that $\En(\psi)\geq\En(\psi^*)$, with  equality if and only if $|\psi|=\psi^*$.

If now we take $\phi$ to be a minimizer, we then immediately obtain $|\phi|=\phi^*$. Moreover, by  the previous Lemma, $|\phi|\in C^{\infty}(B_R)$ and $|\phi|>0$. It remains to show that $\phi=e^{i\theta} |\phi|$. This follows from the fact that both $\phi$ and $|\phi|$ are eigenfunctions of the Schr\"odinger operator $-\Delta -2V_{\phi}$. Since the latter function is strictly positive, $e_{\phi}$ must be the ground state energy of this operator, and is a simple eigenvalue. 
\end{proof}

\section{Uniqueness of Minimizers}
\label{sect3}
In the previous section we have shown that any minimizer, up to a multiplication by a constant phase, must be real, strictly positive, $C^{\infty}$ and radial. To show uniqueness of minimizers it is then sufficient to show uniqueness among functions with these properties. The big advantage of this restriction, as already utilized in \cite{lieb1977existence}, is that the Euler--Lagrange equation for minimizers can be written in the following convenient form.

\begin{rem}
Throughout this paper, we shall make a convenient abuse of notation, and write equivalently $\phi(x)$ or $\phi(r)$ if  $\phi$ is a radial function and $x\in\Rtre$ with $|x|=r$.
\end{rem}

\begin{lem}
\label{ELrad}
Let $\phi \in \Hball$ be a radial function with $\|\phi\|_{2}=1$. Then $\phi$ satisfies Eq.~\eqref{ELeq} if and only if $\phi$ satisfies
\begin{equation}
\label{RadialELeq}
\left[-\frac{d^2}{dr^2} - \frac 2 r \frac d {dr} + 2U_{\phi}(r)\right]\phi(r)=\nu_{\phi} \phi(r)\,,
\end{equation}
where
\begin{itemize}
	\item $U_{\phi}(r):=\int_0^r K(r,s)|\phi(s)|^2 ds$, with $K(r,s)=4\pi s^2(\frac 1 s-\frac 1 r)\geq 0$ for $s\leq r$,
	\item $\nu_{\phi}=e_{\phi}+2I(\phi)-\frac 2 R>0$, with ${I(\phi):=\int_{B_R} \frac {|\phi(x)|^2}{|x|} dx}$.
\end{itemize}
\end{lem}

\begin{proof}
The proof of this Lemma is just a straightforward application of Newton's Theorem \cite[Thm.~9.7]{lieb2001analysis} to the nonlocal term $V_{\phi}$. Indeed, with $r=|x|$ we have
\begin{align}\nonumber
V_{\phi}(x)&=\int_{B_R} \frac{|\phi(y)|^2}{|x-y|}dy - \int_{B_R} \frac {|\phi(y)|^2}{\big| \frac{|y|}{R} x- \frac{R}{|y|} y\big|} dy=\\ \nonumber
&=\frac 1 r \int_{B_r} |\phi(y)|^2 dy +\int_{B_R\setminus B_r} \frac{|\phi(y)|^2}{|y|} dy-\frac 1 R\\
& =-U_{\phi}(r)+I(\phi)-\frac 1 R.
\end{align}
Recalling the original form of the Euler--Lagrange equation \eqref{ELeq}, this identity 
 immediately implies our claim. To show $\nu_{\phi}>0$ one just needs to integrate the equation against $\phi$  and use the positivity of $U_{\phi}$ and of $-\Delta_{B_R}$.
\end{proof}

It is important to note that the nonlocal term $U_{\phi}(x)$  only depends, at a fixed $x$, on the values of $\phi$  on $B_{|x|}$ and \emph{not} on the whole ball $B_R$. By using ODE techniques, as in \cite{lieb1977existence,lenzmann2009uniqueness} (see also \cite{tod1999analytical}), this will allow us to conclude uniqueness of solutions.

\begin{theo}[Uniqueness of minimizers]
\label{uniqueness}
For any $R>0$, there exists a unique positive and $L^2$-normalized minimizer of $\En$.
\end{theo}

\begin{proof}
From Lemma \ref{regularity} and Proposition \ref{radiality} we deduce that any positive minimizer is in $C^{\infty}(B_R)$, is radially decreasing and strictly positive. Moreover, by the previous Lemma, it  satisfies \eqref{RadialELeq}.  Suppose that $\phi_1$ and $\phi_2$ are two distinct positive $L^2$-normalized minimizers. We distinguish two cases: $\nu_{\phi_1}$ and $\nu_{\phi_2}$ can either be equal (first case) or not (second case). 

\textit{First case:} Note that  $\phi_i'(0)=0$ for $i\in\{1,2\}$,  since $\phi_i$ is smooth and radial. If $\phi_1(0)=\phi_2(0)$ it follows from standard fixed point arguments (explained for completeness in Appendix~\ref{appA}) that $\phi_1=\phi_2$ on $B_R$. W.l.o.g. we can hence suppose that $\phi_1(0)>\phi_2(0)$. By integrating the Euler--Lagrange equation using that $\phi_i'(0)=0$, we find
\begin{equation}
\label{ratioest}
\left(\frac {\phi_1}{\phi_2}\right)'(r)=\frac 2 {r^2\phi_2^2(r)}\int_0^r s^2\phi_1(s)\phi_2(s)\left[U_{\phi_1}(s)-U_{\phi_2}(s)\right]ds.
\end{equation}
Exploiting the fact that $U_{\phi}(s)$ only depends on the values of $\phi$ in $[0,s)$, and it does so  \emph{monotonically}, we conclude that if $\phi_1>\phi_2$ on $[0,t)$ for some $t>0$, then $(\phi_1/\phi_2)'(t)>0$. This readily implies that $\phi_1>\phi_2$ on $B_R$, which is  a contradiction to our assumption that both functions are $L^2$-normalized.  

\textit{Second case:} W.l.o.g. we assume that $\nu_{\phi_1}>\nu_{\phi_2}>0$. 
Let $\lambda=\sqrt{\nu_{\phi_1}/\nu_{\phi_2}}>1$ and consider the function $\tilde{\phi_2}(x):=\lambda^2\phi_2(\lambda x)$ defined on $B_{R/\lambda}\subset B_R$. Its $L^2$-norm equals $\sqrt{\lambda}>1$ and it  satisfies
\begin{equation}
\left[-\frac{d^2}{dr^2} - \frac 2 r \frac d {dr} + 2U_{\tilde{\phi_2}}(r)\right]\tilde{\phi_2}(r)=\lambda^2\nu_{\phi_2} \tilde{\phi_2}(r)=\nu_{\phi_1} \tilde{\phi_2}(r)
\end{equation}   
on $B_{R/\lambda}$. 
Hence $\phi_1$ and $\tilde{\phi_2}$ satisfy the equation with same eigenvalue on $B_{R/\lambda}$ and we have reduced the problem to the first case. In particular, we have that either $\phi_1>\tilde{\phi_2}$ or $\phi_1<\tilde{\phi_2}$ or $\phi_1=\tilde{\phi_2}$ on the whole of $B_{R/\lambda}$. Each of these possibilities yields a contradiction, however, since $\tilde{\phi_2}$ has $L^2$-norm strictly larger than $\phi_1$ and is supported on a smaller ball. 
\end{proof}

In combination with Prop.~\ref{radiality}, Thm.~\ref{uniqueness} proves Thm.~\ref{minimizer}.

The unique positive minimizer will henceforth be denoted by $\phi_R$. 
It is natural to expect that, as $R\to \infty$, it converges  to a minimizer of the problem on the full space $\Rtre$. This is indeed the case, as detailed in Appendix~\ref{appB}. 

While the proof of existence of minimizers extends to general domains in $\Rtre$, 
as discussed in Remark \ref{existencealldomains}, the proof of uniqueness relies heavily on  symmetric decreasing rearrangement and hence cannot be easily generalized.  Extending the uniqueness result to more general domains is hence an open problem. As the following counterexample  shows, uniqueness can actually fail on particular domains. Nevertheless, we believe that uniqueness holds \emph{generically}, in the sense that if $\Omega$ is any domain for which different minimizers exist, then a generic perturbation of $\Omega$ should still lead to a unique minimizer (up to phase). We conjecture that convexity of $\Omega$ is a sufficient condition to ensure uniqueness. 

\begin{rem}
\label{couterexample}
Consider two disjoint balls of the same size in $\Rtre$, $B_1:=B_R(x_1)$ and $B_2:=B_R(x_2)$, with $|x_1-x_2|>2R$. Let $\Om=B_1\cup B_2$ and consider the Pekar functional defined on $\Om$:
\begin{equation}
\mathscr{E}_{\Om}(\phi)=\int_{\Om} |\nabla \phi|^2 dx - 4\pi\int_{\Om}\int_{\Om} (-\Delta_{\Om})^{-1} (x,y) |\phi(x)|^2 |\phi(y)|^2 dx dy=T_{\Omega}(\phi)-W_{\Omega}(\phi).
\end{equation}
Here $(-\Delta_{\Om})^{-1} (x,y)$ denotes  the integral kernel of the inverse Dirichlet Laplacian on $\Omega$. Any  $L^2$-normalized $\phi\in H^1_0(\Om)$ can be written as $\phi=\sqrt{t} \phi_1+\sqrt{1-t} \phi_2$, for some $t\in [0,1]$ and $L^2$-normalized $\phi_1\in H^1_0(B_1)$, $\phi_2\in H^1_0(B_2)$.  For general functions $f_1,f_2$, 
\begin{align}\nonumber
&\expval{-\Delta_{\Omega}^{-1}}{tf_1+(1-t)f_2}\\
&= t \expval{-\Delta_{\Omega}^{-1}}{f_1}+(1-t)\expval{-\Delta_{\Omega}^{-1}}{f_2}-t(1-t)\expval{-\Delta_{\Omega}^{-1}}{f_1-f_2}.
\end{align} 
By the positivity of $-\Delta_{\Omega}^{-1}$ as an operator, the last term  is strictly negative unless $t\in\{0,1\}$ or $f_1=f_2$. In other words, $\expval{-\Delta_{\Omega}^{-1}}{\cdot}$ is strictly convex, which holds true for general $\Omega$, in fact. 
In particular
\begin{align}\nonumber
\mathscr{E}_{\Om}(\phi)&=t \int_{B_1} |\nabla \phi_1|^2 dx + (1-t) \int_{B_2} |\nabla \phi_2|^2 dx -W_{\Omega}\left(\sqrt{t}\phi_1+\sqrt{1-t} \phi_2\right) \\
&\geq t\mathscr{E}_{B_1}(\phi_1)+(1-t) \mathscr{E}_{B_2}(\phi_2) \geq E_R 
\end{align}
and the first inequality is strict unless $t=0$ or $t=1$. We conclude that any minimizer of $\mathscr{E}_\Omega$ is obtained by translating  a minimizer of $\En$ by $x_1$ \emph{or} $x_2$. In particular, uniqueness up to phase does not hold on $\Om$. 

The fact that $\Om$ has two distinct connected components is not essential in our argument. The lack of uniqueness would still hold, by continuity, if $B_1$ and $B_2$ were connected by a sufficiently narrow corridor, respecting the symmetry between the two balls.  On the other hand,  a \emph{generic} perturbation of $\Om$ (or of $\Om$ connected by a corridor) would restore uniqueness up to phase of minimizers, since it would break the symmetry.
\end{rem}

\section{Study of the Hessian}
\label{sect4}

Recall that for given $R>0$, $\minP$ denotes the unique $L^2$-normalized positive minimizer of $\En$ on $B_R$. In this section we study the Hessian of $\En$ at $\minP$, following ideas in~\cite{lenzmann2009uniqueness} (see also \cite{wei2009strongly}). Let $\phi$ be any function in $\Hball$. A straightforward computation shows that 
\begin{equation}
\En\left(\frac{\minP+\varepsilon\phi}{\|\minP+\varepsilon\phi\|_2}\right)=\En(\minP)+\varepsilon^2 \HP(\phi)+O(\varepsilon^3)
\end{equation}
as $\varepsilon\to 0$, where
\begin{equation}
\label{Hessian}
\HP(\phi)=\expval{L_{-}}{\Im(\phi)}+\expval{QL_+Q}{\Re(\phi)},
\end{equation}
$Q=\id- |\phi_R\rangle\langle\phi_R|$, and the operators $L_\pm$ are given by
\begin{equation}\label{def:lpm}
L_- := - \Delta_{B_R} - 2V_{\minP}-e_{\minP} \ , \quad L_+ = L_- - 4 X
\end{equation}
with 
\begin{equation}
(Xf)(x):=4\pi\minP(x)\int_{B_R}\Gker \minP(y)f(y) dy\,.
\end{equation}
We recall that $e_{\minP}=T_R(\minP)-2W_R(\minP)$. Moreover $V_\phi$ is defined in \eqref{def:Vp}. 
Since $\phi_R$ is smooth, it is not difficult to see that both $V_{\minP}$ and $X$ are bounded operators. In particular, the domain of $L_\pm$ equals the domain of $\Delta_{B_R}$, namely $H^2(B_R)\cap H^1_0(B_R)$.
Using \eqref{DGF}, we find it convenient to decompose $X$ as $X=X_1 - X_2$ with
\begin{equation}\label{def:x12}
(X_1f)(x):=\minP(x)\int_{B_R}\frac {\minP(y)f(y)}{|x-y|} dy \ , \quad (X_2f)(x):=\minP(x) \int_{B_R}\frac{\minP(y)f(y)}{\big| \frac{|y|}{R} x- \frac{R}{|y|} y\big|} dy\,.
\end{equation}

Note that $\minP\in \ker \Lminus$ by the Euler--Lagrange equation \eqref{ELeq}. Since $Q\phi_R =0$, clearly also $\minP\in \ker Q\Lplus Q$. Our aim is to show that $0$ is a simple eigenvalue for both $\QLplus$ and $\Lminus$. This will imply the strict positivity of the Hessian on $\ran Q$. Indeed, by minimality of $\minP$, both operators are non-negative and, since the domain under consideration is  bounded, have  compact resolvents and discrete spectrum. 

The simplicity of $0$ as an eigenvalue of $\Lminus$ follows from the fact that  $\Lminus$ is a Schr\"odinger operator with $\infspec \Lminus=0$ (since the corresponding eigenfunction $\minP$ is positive). 
Note that the non-triviality of $\ker \Lminus$ is a consequence of the $U(1)$-symmetry of $\En$ leading to uniqueness \emph{up to phase} of the minimizer only. Indeed, purely imaginary perturbations of $\minP$ by functions in $\spn\{\minP\}$  correspond to phase rotations of $\minP$.

The analysis of  $\ker\QLplus$ is more tricky.  The presence of the projection $Q$ does not allow the use of standard arguments   to show simplicity of the least eigenvalue based  on positivity. It will be essential to utilize that 
$\Lplus$ commutes with rotations. We recall that
\begin{equation}
L^2(B_R)=\bigoplus_{l=0}^{\infty} \mathscr{H}_l,
\end{equation}
where $\mathscr{H}_l:=L^2([0,R],r^2dr)\otimes \mathscr{Y}_l$, $\mathscr{Y}_l=\spn\{Y_{lm}\}_{m=-l}^{l}$ is the $(2l+1)$-dimensional eigenspace corresponding to the eigenvalue $l(l+1)$ of the negative spherical Laplacian on $L^2(\mathbb{S}^2)$ and $Y_{lm}$ is the $m$-th spherical harmonic of angular momentum $l$. The fact that $\Lplus$ commutes with rotations implies that $\Lplus$ acts invariantly on each $\mathscr{H}_l$, i.e., it can be decomposed as
\begin{equation}
\Lplus=\bigoplus_{l=0}^{\infty} \Lplus\restr{\mathscr{H}_l}=:\bigoplus_{l=0}^{\infty} \Lplus^{(l)}\,.
\end{equation}
Since  $\minP$ is radial, also $Q$ leaves each $\mathscr{H}_l$ invariant (in particular $Q\restr{\mathscr{H}_l}=\id$ if $l\geq1$), hence 
\begin{equation}
\QLplus=\bigoplus_{l=0}^{\infty}(Q \Lplus Q)\restr{\mathscr{H}_l}=\left(QL_{+}^{(0)}Q \right)\oplus \left(\bigoplus_{l=1}^{\infty} \Lplus^{(l)}\right)\,.
\end{equation}
Identifying the kernel of $\QLplus$ is equivalent to identifying the kernels of $Q\Lplus^{(0)}Q\restr{\mathscr{H}_0}$ and of $\Lplus^{(l)}$ for $l\geq 1$. We start with the study of $Q\Lplus^{(0)}Q$, the only operator in which $Q$ still appears, complicating the analysis. The operators $\Lplus^{(l)}$, in which $Q$ does not appear, will be studied with more standard arguments below.

\begin{prop}\label{kerLo}
\begin{equation}
\ker \left(Q\Lplus^{(0)}Q\right)=\ker Q=\spn\{\minP\}.
\end{equation}
\end{prop}

\begin{proof}
Since $\ker(Q\Lplus^{(0)})\cap \ran Q=\{0\}$ implies $\ker Q\Lplus^{(0)}Q\restr{\mathscr{H}_0}=\ker Q$,  our strategy will be to show that $\ker(Q\Lplus^{(0)})$  does not contain any non-null functions that are in $\ran Q$. Since all operators are real (i.e., commute with complex conjugation), it is sufficient to consider real-valued functions. 
We consider a $f\in \dom\Lplus^{(0)}$ (which in particular implies $f\in \mathscr{H}_0$, i.e., $f$ radial) and observe that, by Newton's  Theorem, 
\begin{equation}
(\Lplus f)(r)=(\Lpplus f)(r)-\sigma(f)\minP(r),
\end{equation}
with
\begin{equation}\begin{aligned}
\label{Lppesigma}
&(\Lpplus f)(r):=(\Lminus f)(r)+4\minP(r)\int_{B_r} \left(\frac 1 {|y|}-\frac 1 r\right)\minP(y) f(y) dy,\\
& \sigma(f):=4 \int_{B_R} \left(\frac 1 {|y|}-\frac 1 R\right)\minP(y) f(y) dy \,.
\end{aligned}\end{equation}
Any $f\in\dom \Lplus^{(0)}$ is in $\ker (Q\Lplus^{(0)})$ if and only $\Lplus f=\lambda \minP$ for some $\lambda\in \mathbb{R}$ and, by the previous discussion, this is true if and only if
\begin{equation}
\label{equation1}
\Lpplus f=\mu \minP \hspace{2mm} \text{for some} \hspace{2mm} \mu\in \mathbb{R}. 
\end{equation}

The operator $\Lpplus$ can be naturally defined on the extended domain $H^2(B_R)$ (without Dirichlet boundary conditions at $R$) and it will be convenient to do so in the following. 
From the above discussion we infer that  $f\in\ker(Q\Lplus^{(0)})$ must be of the form $f=v+c\varphi$, with $c\in \R$,  $v$  a solution of $\Lpplus v=0$ and $\varphi$ being a \emph{particular} solution of \eqref{equation1}, with $\mu\neq 0$. While $f$ needs to satisfy Dirichlet boundary conditions, i.e., $f\in H^1_0(B_R)$, this need not be the case for $v$ and $\varphi$ separately, however. 
In the following,  we will exhibit a particular solution $\varphi$ that is \emph{radial}, hence we are only interested in radial solutions of $\Lpplus v=0$.
	
We begin by studying the radial solutions of $\Lpplus v=0$. A bootstrapping argument shows that any such $v$ must be  in $C^{\infty}(B_R)$. Moreover, by Newton's Theorem, $v$ satisfies
\begin{equation}
v''(r)+\frac 2 r v'(r)=a(r)v(r)+b(r),
\end{equation}
where
\begin{equation}
a(r):=-2V_{\minP}(r)-e_{\minP} \  , \quad
b(r):=4\minP(r)\int_{B_r} \left(\frac 1 {|y|}-\frac 1 r\right)\minP(y) v(y) dy\,.
\end{equation}
By the regularity of $v$, we have $v'(0)=0$. Arguing as in the proof of Lemma \ref{localuniqueness}, we see  that the equation possesses no non-trivial  solution that vanishes at the origin.
	
Recall that $\minP$ satisfies  
\begin{equation}
\minP''(r)+\frac 2 r \minP'(r)=a(r)\minP(r)\,.
\end{equation}
By applying the same computations as in the proof of Thm.~\ref{uniqueness}, using $v'(0)=\minP'(0)=0$, we obtain
\begin{equation}
\left(\frac{v}{\minP}\right)'(r)=\frac 1 {r^2\minP^2(r)} \int_0^r s^2b(s)\minP(s) ds\,.
\end{equation} 
Note that $b(r)\geq0$ if $v\geq0 $ in $[0,r)$. Assuming that  $v(0)>\minP(0)$ this implies that $v>\minP$ on $B_R$. In other words, any non-trivial radial solution of $\Lpplus v=0$ has a multiple which is strictly larger than $\minP$ on $B_R$. In particular any non-trivial radial solution must have constant sign.
	
Consider now the radial function $\varphi(r):=2\minP(r)+r \minP'(r)$. We observe that $\varphi\not\in \ran Q$, since  $\bra{\minP}\ket{\varphi}=1/2$ as an argument using integration by parts shows. A straightforward computation shows that $\Lplus \varphi= \lambda \minP$ for some $\lambda\in \mathbb{R}$, which implies that also 
 $\Lpplus \varphi=\mu \minP$ for some $\mu\in\R$. We claim that  $\mu\neq0$, which is an immediate consequence of our previous findings about radial solutions of $\Lpplus v=0$. Indeed, $\varphi(0)>0$ whereas $\varphi(R)<0$ (a proof of this last statement is given in Lemma \ref{signofder} in Appendix~\ref{appA}), hence $\varphi$ does not have constant sign and cannot be in $\ker\Lpplus$. We  conclude that $\varphi$ is a particular solution of \eqref{equation1} and this implies, by the previous discussion, that any $f \in \ker(Q\Lplus^{(0)})$ must be of the form $f=v+c\varphi$, for some $v\in \ker \Lpplus$ and some $c\in \mathbb{R}$. The case $v\equiv 0$ immediately yields $f=0$, since $\varphi$ does not satisfy the boundary condition $f(R)=0$. All the other solutions $v$ have constant sign, thus the boundary condition $f(R)=0$ is satisfied if and only if $c$ has the same sign of $v$. In particular, 
\begin{equation}
\bra{\minP}\ket{f}=\bra{\minP}\ket{v}+\frac c 2 \neq 0
\end{equation}
unless $f=0$, i.e., $f\in \ran Q$ if and only if $f=0$. We conclude  that $\ker Q\Lplus^{(0)}\cap\ran Q=\{0\}$, as claimed.
\end{proof}

We now proceed with the study of $\ker \Lplus^{(l)}$ for $l\geq 1$. We first investigate the explicit expressions of these operators. 
We note that the action of $\Lplus$ is not only invariant on $\mathscr{H}_l=L^2([0,R],r^2dr)\otimes \mathscr{Y}_{l}$, but it also acts as the identity on the second factor. Hence we can identify the operators $\Lplus^{(l)}$ with operators acting  on $L^2([0,R],r^2dr)$ only, which we will denote by the same symbol  for simplicity. 
That is, if $\phi \in \mathscr{H}_l$ is of the form $\phi(r\omega)= \sum\nolimits_{m=-l}^{m=l} \phi_{m}(r) Y_{ml}(\omega)$ for $\omega\in \mathbb{S}^2$, then 
\begin{equation}
\Lplus\phi=\Lplus^{(l)}\phi= \sum_{m=-l}^{m=l} (\LplusL\phi_{m}) Y_{ml}\,,
\end{equation}
where he operators $\LplusL$  are  defined on $L^2([0,R],r^2dr)$ by
\begin{equation}
\LminusL=-\frac{d^2}{dr^2} - \frac 2 r \frac d {dr} +\frac {l(l+1)}{r^2}-e_{\minP}-2V_{\minP}
\end{equation}
and $\LplusL=\LminusL-4\XL$ with $\XL=\XLone-\XLtwo$, where 
\begin{equation}\label{x1x2l}\begin{aligned}
&(\XLone\phi)(r)=\frac{4\pi}{2l+1}\minP(r)\int_0^R\phi(s)\minP(s)s^2\frac{\min\{r,s\}^l}{\max\{r,s\}^{l+1}}ds,\\
&(\XLtwo\phi)(r)=\frac{4\pi}{2l+1}\minP(r)\int_0^R\phi(s)\minP(s)s^2\frac{(rs)^l}{R^{2l+1}} ds.
\end{aligned}\end{equation}
This follows from a straightforward computation, using the \emph{multipole expansion} (see, for example \cite{condon1951theory})
\begin{equation}
\frac 1 {|x-y|}=4\pi\sum_{k=0}^{\infty}\sum_{n=-k}^{k} \frac 1 {2k+1} \frac{\min\{|x|,|y|\}^k}{\max\{|x|,|y|\}^{k+1}} Y_{kn}(\omega_x)Y^*_{kn}(\omega_y)\,.
\end{equation}

Let us define the operator $\tLplus:= \Lminus-4X_1$, and the corresponding restriction to $\mathscr{H}_l$, $\tLplusL:=\LminusL-4\XLone$. 

\begin{lem}
\label{lemgoodproperties}
For $l\geq 1$ the operators $\LplusL$ and $\tLplusL$ satisfy the Perron--Frobenius property, i.e., their least eigenvalue is \emph{simple} and there exists a corresponding eigenfunction which is \emph{strictly} positive on $(0,R)$. This eigenfunction is in $C^{\infty}((0,R))$ and has strictly negative (left) derivative at $r=R$.
\end{lem}

\begin{proof}
We will give the proof for the operators $\LplusL$; it will be important that $\XL=\XLone-\XLtwo$ is positivity improving, which can be checked easily using the explicit form \eqref{x1x2l}.  The proof for $\tLplusL$ works in exactly the same way, using simply that $X_1^{(l)}$ is positivity improving instead.

It will be convenient to introduce the unitary and positive transformation
\begin{equation}\label{def:U}
U:L^2([0,R],r^2dr)\to L^2([0,R],dr) \hspace{5mm}\text{with}\hspace{5mm} (Uf)(r)=rf(r)
\end{equation}
which satisfies 
\begin{equation}\begin{aligned}
&U\LplusL U^{-1} = -\frac {d^2}{dr^2}+\frac {l(l+1)}{r^2}+V,\\
&V:=-e_{\minP}-2V_{\minP}-4UX^{(l)}U^{-1}.
\end{aligned}\end{equation}
Since $U$ is positive, it is equivalent to show the Perron--Frobenius property for $U\LplusL U^{-1}$. 

Since $V$ is bounded, the operators $U\LplusL U^{-1}$ have compact resolvent and eigenfunctions corresponding to the least eigenvalue certainly exist. 
By bootstrapping, we conclude that they are $C^{\infty}((0,R))$. Moreover, if $\phi\geq 0$ is such an eigenfunction, then $\phi>0$ on $(0,R)$. Indeed, if we suppose  that $\phi$ is not strictly positive, then there exists an $r_0\in(0,R)$ such that $\phi(r_0)=0$. Evaluating the Euler--Lagrange equation at $r_0$ we find,  using that $U$ is positive and $X^{(l)}$ is positivity improving,
\begin{equation}
-\phi''(r_0)=4(UX^{(l)} U^{-1} \phi)(r_0)>0.
\end{equation}
This is clearly a contradiction since $\phi$ attains a minimum in $r_0$. From this, we can conclude by standard arguments that the Perron--Frobenius property holds.  

Finally, we need to show that $\phi'(R)<0$ if $\phi$ is the positive ground state function. We already know that $\phi(R)=0$ and $\phi'(R)\leq 0$ (since $\phi$ is positive). If by contradiction $\phi'(R)=0$  standard uniqueness arguments along the lines of Lemma \ref{signofder} imply that $\phi\equiv 0$. Note that also this property is preserved by $U$ since $\phi(R)=0$. 
\end{proof}

\begin{rem}\label{rem2}
From this Lemma and the fact that $X_2^{(l)}$ is positivity improving we conclude that for each $l\geq 1$ we have $\infspec \LplusL> \infspec \tLplusL$. Thus, in order to show that $\ker \LplusL=\{0\}$  for $l\geq 1$, it is sufficient to show $\infspec \tLplusL\geq 0$ for $l\geq 1$. It is actually even possible to show $\infspec \tLplusL> 0$ for $l\geq 1$, which is the content of the next Proposition. This is going to be relevant for Remark~\ref{alsootherkernel} at the end of this section. 
\end{rem}

\begin{prop}
\label{kerLpl}
For any $l>1$, we have
\begin{equation}\label{LLP}
\infspec \tLplusL> \infspec \tilde{L}_+^{(1)} >0.
\end{equation}
In particular, $\ker \LplusL =\{0\}=\ker \tilde{L}_+^{(l)}$ for all $l\geq 1$. 
\end{prop}

\begin{proof}
For $i\in\{1,2,3\}$, we have $\dxi\minP(x)\in \mathscr{H}_1$, with
\begin{equation}
\dxi\minP(x)=\minP'(r) \frac {x_i}{r}=\sum_{m=-1}^1 c_m^i \minP'(r) Y_{1m}(\omega)
\end{equation}
for suitable $c_m^i$. 
Since $\minP'(R)<0$, this function is not in the domain of $\tLplus^{(1)}$. As in the proof of Prop.~\ref{kerLo}, we can consider the extension of $\tLplus^{(1)}$ to $H^2(B_R)\cap \mathscr{H}_l$, however (ignoring the Dirichlet boundary condition). 
A straightforward computation shows that 
\begin{equation}
\tLplus^{(1)}  \minP' =0  
\end{equation}
 i.e., $\minP'$ is in the kernel of the extended operator. 
	
Let $\phi$ denote the unique positive ground state of the original, unextended $ \tilde{L}_+^{(1)}$, with ground state energy $\tilde e_1$. The function $\phi$ is strictly positive on $(0,R)$ and satisfies $\phi'(R)<0$. 
Integrating by parts, we have 
\begin{align}\nonumber
0&=\bra{ \phi}\ket{\tilde{L}_+^{(1)} \minP'}
=\bra{ \tilde{L}_+^{(1)} \phi}\ket{\minP'}
+\phi'(R)\minP'(R)R^2-\phi(R)\minP''(R)R^2=\\
&= \tilde e_1 \bra{\phi}\ket{\minP'}+\phi'(R)\minP'(R)R^2.
\end{align}
In particular, we  conclude that
\begin{equation}
\tilde e_1 =-\frac{\phi'(R)\minP'(R)R^2}{\bra{\phi}\ket{\minP'}}> 0\,,
\end{equation}
which is the second inequality in \eqref{LLP}.  For the first inequality, observe that if $0<\phi \in L^2([0,R],r^2dr)$ and $l\geq 2$, 
\begin{align}\nonumber
&(\tLplusL \phi)(r) - (\tilde{L}_+^{(1)}\phi)(r)=\left(\frac{l(l+1)}{r^2} -\frac 2 {r^2}\right) \phi(r)  \\
& +4\pi \minP(r) \int_0^R \phi(s)\minP(s)s^2\left(\frac {\min\{r,s\}}{3\max\{r,s\}^2}-\frac {\min\{r,s\}^l}{(2l+1)\max\{r,s\}^{l+1}}\right)ds>0.
\end{align}
By Lemma \ref{lemgoodproperties}, the ground state $\phi_l$ of $\tLplusL$ is strictly positive. Thus
\begin{equation}
\infspec \tLplusL=\expval{\tLplusL}{\phi_l}_{L^2([0,R],r^2dr)}>\expval{\tilde{L}_+^{(1)}}{\phi_l}_{L^2([0,R],r^2dr)}\geq \tilde e_1 >0,
\end{equation}
which completes the proof.
\end{proof}

With the aid of Propositions~\ref{kerLo} and~\ref{kerLpl}, we can now give the proof 
of Theorem \ref{coercivitythm}. The proof  follows closely \cite[Appendix~A]{frank2019quantum}, with some minor modifications due to the fact that our statement is slightly stronger than the one in \cite{frank2019quantum}. We emphasize that the hard part of the proof was  establishing the triviality of the kernel of $QL_+Q$  (which enters as an assumption in \cite{frank2019quantum}), the remaining part uses  only fairly standard arguments.

\begin{proof}[Proof of Theorem \ref{coercivitythm}]
We shall actually prove the following slightly stronger inequality: For any $L^2$-normalized $\phi \in \Hball$ with $\langle \phi|\phi_R\rangle \geq 0$, 
\begin{equation}\label{svo}
\En(\phi)\geq \En(\phi_R)+K_R  \int_{B_R} |\nabla (\phi_R-\phi)|^2 dx 
\end{equation}
for some $K_R>0$ (independent of $\phi$). Because of the invariance of $\En(\phi)$ under multiplication of $\phi$ by a complex phase, \eqref{svo} readily implies~\eqref{eq:thm2}. 

To show~\eqref{svo} we shall proceed in two steps, one to ensure that the estimate holds \emph{locally} and one to ensure that it holds \emph{globally}.

\textit{Step 1}: In this step we show that~\eqref{svo} holds locally. Let $\phi\in \Hball$ with $\|\phi\|_{2}=1$ and $\langle\phi|\phi_R\rangle \geq 0$. Denoting $\delta=\phi-\minP$ and expanding $\En$ around $\minP$, we have
\begin{align}\nonumber
\En(\phi)&=\En\left(\minP+\delta \right)\\
&=\En(\minP)+\expval{\Lminus}{\Im\delta}+\expval{\Lplus}{\Re\delta}+O(\|\delta\|_{H^1(B_R)}^3)
\end{align}
for small $\|\delta\|_{H^1(B_R)}^3$, with $L_\pm$ defined in~\eqref{def:lpm}. Recall that $L_- = Q L_- Q$ for $Q=\id -|\phi_R\rangle\langle\phi_R|$, and that $L_+ = L_--4X$. 

In order to utilize the previous results, we would need $QXQ$ in place of $X$. To estimate the difference, observe that, since both $\minP$ and $\phi$ have $L^2$-norm equal to $1$, we have
\begin{equation}\begin{aligned}
&\|\delta\|_2^2=2-2 \bra{\minP}\ket{\phi},\\
&(\id-Q)\Re \delta=\minP \bra{\minP}\ket{\delta}=\minP\left( \bra{\minP}\ket{\phi}-1\right)=-\minP\frac{\|\delta\|_2^2}{2}.
\end{aligned}\end{equation}
This readily implies that 
\begin{equation}
\expval{X}{\Re \delta}=\expval{QXQ}{\Re \delta}+O(\|\delta\|_{H^1(B_R)}^3).
\end{equation}
In particular, we have
\begin{equation}
\label{firstequality}
\En(\phi)=\En(\minP)+\expval{\Lminus}{\Im\delta}+\expval{\Lplus}{Q\Re\delta}+O(\|\delta\|_{H^1(B_R)}^3).
\end{equation}
As argued in the beginning of this section, we have $L_- \geq \kappa_- Q$ for some $\kappa_->0$. Moreover,  Propositions~\ref{kerLo} and~\ref{kerLpl} imply that $Q L_+ Q \geq \kappa_+ Q$ for some $\kappa_+>0$. 
With $\kappa=\min\{\kappa_-, \kappa_+\}>0$, we thus have 
\begin{align}\nonumber
&\expval{\Lminus}{\Im\delta}+\expval{\Lplus}{Q\Re\delta}\\
&\geq\kappa(\|Q\Im \delta\|_2^2+\|Q\Re \delta\|_2^2) =\kappa \|Q\delta\|_2^2\,.
\end{align}
The assumption $\langle\phi|\phi_R\rangle\geq 0$ implies that 
\begin{equation}
\|Q\delta\|^2_2 = \|\delta\|^2_2 - \langle \delta|\phi_R\rangle^2 =  \|\delta\|^2_2  \left( 1 -\tfrac 14 \|\delta\|^2_2\right) \geq \tfrac 12 \|\delta\|^2_2
\end{equation}
and hence
\begin{equation}
\label{ineqone}
\expval{\Lminus}{\Im\delta}+\expval{\QLplus}{\Re\delta}\geq
 \frac {\kappa} 2 \|\delta\|_2^2.
\end{equation}

Next we want to improve this lower bound by including the full $H^1$-norm of $\delta$. We can do this by exploiting the explicit form of $\Lplus$ and $\Lminus$. Indeed, by the boundedness of $V_{\minP}$, 
\begin{equation}
\Lminus\geq -\Delta - C\,.
\end{equation}  
Using the smoothness of $\minP$, it not difficult to see that also
\begin{equation}
\QLplus\geq -\Delta -C \,. 
\end{equation} 
In particular, 
\begin{equation}
\label{ineqtwo}
\expval{\Lminus}{\Im\delta}+\expval{\QLplus}{\Re\delta}\geq \expval {-\Delta-C}{\delta}=\|\nabla \delta\|_2^2-C\|\delta\|_2^2.
\end{equation}
By interpolating between \eqref{ineqone} and \eqref{ineqtwo}, we have
\begin{equation}
\expval{\Lminus}{\Im\delta}+\expval{\QLplus}{\Re\delta}\geq \left[ \frac {\kappa(1-\alpha)} 2-C\alpha \right] \|\delta\|_2^2+ \alpha \|\nabla \delta\|_2^2
\end{equation}
for any $0\leq \alpha\leq 1$. 
By choosing $\alpha=\frac {\kappa}{2+\kappa+2C_3}$ and substituting in \eqref{firstequality}, we obtain
\begin{equation}
\label{almostrighteq}
\En(\phi)\geq\En(\minP)+\frac{\kappa}{2+\kappa+2C_3} \|\nabla\delta\|^2_2 +O(\|\delta\|_{H^1(B_R)}^3).
\end{equation}
In particular,  there exist $c>0$ and $K>0$ such that, if $\|\delta\|_{H^1(B_R)}\leq c$, then
\begin{equation}
\En(\phi)\geq\En(\minP)+K \|\nabla(\phi-\phi_R)\|^2_{2}\,.
\end{equation}
In words, we have shown that the desired coercivity estimate holds locally, in the sense that it holds whenever the $H^1$-norm of $\delta$ is sufficiently small.  

\textit{Step 2:} Suppose by contradiction that we cannot find a $K_R$ such that~\eqref{svo} holds  globally on $\Hball$. Then there exist $\phi_n\in\Hball$ with $\|\phi_n\|_2=1$ and $\langle \phi_n|\phi_R\rangle \geq 0$ 
such that 
\begin{equation}
\En(\phi_n)\leq\En(\minP)+\frac 1 n \|\nabla(\minP-\phi_n)\|_2^2 
\end{equation}
for any $n\in \mathbb{N}$. 
At the same time, we recall that by the estimate \eqref{Energybound}, we have
\begin{equation}
\En(\phi_n)\geq \frac 1 2 \|\nabla \phi_n\|^2_2-C.
\end{equation}
By combining the two inequalities, we see that $\phi_n$ is bounded in $H^1(B_R)$. Thus, also $\|\nabla(\minP-\phi_n)\|_2^2$ is bounded, which implies that $\En(\phi_n)\to \En(\minP)$, i.e., $\phi_n$ is a minimizing sequence. Therefore, up to subsequences, $\phi_n$ is converging in $H^1$ to a minimizer, i.e., to $e^{i\theta}\minP$ for some $\theta\in[0,2\pi)$, by the compactness properties exploited in  the proof of Theorem \ref{existencethm}. (There, only $L^2$-convergence and weak $H^1$-convergence are proved, but the strong $H^1$-convergence follows immediately from the convergence of the individual parts of the functional.) The assumption $\langle \phi_n|\phi_R\rangle\geq 0$ implies that  $\|\phi_n-\minP\|_2\leq \|\phi_n-e^{i\theta}\minP\|_2\to 0$, which in turn  implies that $\theta=0$. Thus, we find a contradiction since $\phi_n\to \minP$ in $H^1$ and we can use the local result of step 1.
\end{proof}

\begin{rem}
As explained in Remark \ref{couterexample}, uniqueness of minimizers may fail on general domains, which implies that also~\eqref{eq:thm2} fails in this case. We still believe the bound to hold \emph{locally} even if uniqueness fails, however.  In other words, the Hessian at the minimizer(s) should be non-degenerate, in which case step $1$ in the previous proof still applies. Uniqueness of minimizers enters only in step $2$. 
\end{rem}

\begin{rem}
\label{alsootherkernel}
As a final remark, we point out that  all the results in this paper can be obtained also if considering, instead of the Pekar functional \eqref{energy} on a ball, the Pekar functional on the full space, restricted to  functions in $\Hball$ (extended by $0$ outside $B_R$), i.e., the functional~\eqref{tenergy}. 
Indeed, existence of minimizers can be shown exactly as in Section~3, as well as regularity of minimizers. To show that minimizers must be radial, one needs to use the strong form of the  Riesz inequality  proved in \cite{lieb1977existence} instead of  Talenti's inequality. 
Note that on radial functions the two functionals $\En$ and $\tEn$ differ only by a constant $1/R$ (by Newton's Theorem), i.e., if $\phi\in \Hball$ is radial and $L^2$-normalized then
\begin{equation}
\En(\phi)=\tEn(\phi)+\frac 1 R.
\end{equation}
In particular, the two functionals have the same minimizers. 

The non-degeneracy results for the Hessian can also be extended to $\tEn$. If we denote by $\tilde{H}_R$ the Hessian of $\tEn$ at $\phi_R$, we have
\begin{equation}
\label{HessDec}
{\tilde{H}_R}(\phi)=\expval{\Lminus}{\Im \phi}+\expval{\tLplus}{Q\Re\phi}.
\end{equation}
Here, $Q=\id- |\phi_R\rangle\langle\phi_R|$ as above,  $L_\pm$ is  defined  in~\eqref{def:lpm}, and  $\tilde L_+ = L_+ - 4X_2 = L_- - 4X_1$ with $X_{1,2}$ defined in~\eqref{def:x12}. 
The decomposition \eqref{HessDec} implies that the study of imaginary perturbations can be carried out as above. For real perturbations, we  can again decompose the Hessian w.r.t. spherical harmonics, and carry out the analysis  in each angular momentum sector separately. For $l=0$, i.e., for radial functions, we can argue exactly as in the proof of Proposition \ref{kerLo}, since the modification of the interaction kernel only affects the term $\sigma$ in~\eqref{Lppesigma}, leaving the operator $\Lpplus$ unchanged. 
For $l\geq 1$, we have actually already shown above that  $\tLplus>0$ on $\mathscr{H}_l$. Also the proof of Theorem \ref{coercivitythm} carries over to the modified interaction kernel without change. We thus  conclude that Theorems~\ref{minimizer} and~\ref{coercivitythm} are also valid, as stated, for the functional $\tEn$. 
\end{rem}

\appendix

\section{Uniqueness Properties for the Radial Euler--Lagrange Equation}\label{appA}

In this Appendix we show two Lemmas dealing with the radial Euler--Lagrange equation~\eqref{RadialELeq}.  The first one proves  uniqueness of solutions with the same boundary conditions at $r=0$.
We recall that $U_{\phi}=4\pi\int_0^r s^2(\frac 1 s-\frac 1 r)|\phi(s)|^2 ds$. We take the eigenvalue $\nu_{\phi}=1$ for simplicity, which can be achieved by a suitable rescaling. 

\begin{lem}\label{localuniqueness}
Let $v_1, v_2 \in C^2([0,T])$ be two  solutions of
\begin{equation}\label{star}
\begin{cases}
&-v''(r)-\frac 2 r v'(r) +2U_v(r) v(r)=v(r) \hspace{10mm} r\in[0,T]\\
&v(0)=a,\\
&v'(0)=0
\end{cases}
\end{equation}
for some $a\in \R$ and  $T>0$. Then $v_1=v_2$ in $[0,T]$.
\end{lem}

\begin{proof}
Let $\sigma_i(r):=rv_i(r)$. Then $\sigma_i'(r)=v_i(r)+rv_i'(r)$ and $\sigma_i''=2U_{v_i} \sigma_i-\sigma_i$. 
By applying Taylor's formula with remainder in integral form,  and denoting $I_r:=[0,r]$, we have 
\begin{align}\nonumber
&|\sigma_1(r)-\sigma_2(r)|\leq \left|\int_0^r \left[ \sigma_1(s)(2U_{v_1}(s)-1) - \sigma_2(s)(2U_{v_2}(s)-1)\right] (r-s)ds\right|\\ \nonumber
&\leq \int_0^r(r-s)|\sigma_1(s)-\sigma_2(s) |ds+2\int_0^r(r-s)|\sigma_1(s)U_{v_1}(s)-\sigma_2(s) U_{v_2}(s)| ds \\
&\leq  r^2\left[ \|\sigma_1-\sigma_2\|_{L^{\infty}(I_r)}(  \tfrac 12 + \|U_{v_1}\|_{L^{\infty}(I_r)})+\|\sigma_2\|_{L^{\infty}(I_r)}\|U_{v_1}-U_{v_2}\|_{L^{\infty}(I_r)}\right]. \label{a2}
\end{align}
Boundedness of $v_{1,2}$ implies that  $\|U_{v_1}\|_{L^{\infty}(I_r)}\leq C r^2$ and  $\|\sigma_2\|_{L^{\infty}(I_r)}\leq C r$ for suitable constants $C$.  Elementary computations also show that $\|U_{v_1}-U_{v_2}\|_{L^{\infty}(I_r)}\leq C  r \|\sigma_1-\sigma_2\|_{L^{\infty}(I_r)}$. In particular, from~\eqref{a2} we conclude that 
\begin{equation}
\|\sigma_1-\sigma_2\|_{L^{\infty}(I_r)}\leq r^2(\tfrac 12 + C r^2 ) \|\sigma_1-\sigma_2\|_{L^{\infty}(I_r)}.
\end{equation}
Thus, $\sigma_1=\sigma_2$ on $I_{\delta}$ whenever $\delta^2(\tfrac 12+ C \delta^2)<1$, and we have  local uniqueness of solutions for~\eqref{star}. The same computations can be carried out \emph{mutatis mutandis} by considering an arbitrary starting point instead of $0$. In particular, we can go from local uniqueness to global uniqueness by iteration of the argument: if the two functions only coincide in a \emph{maximal} interval $[0,T^*]$ with $T^*<T$ (note that by continuity they necessarily coincide on a closed interval) then we get a contradiction by applying the argument with starting point $T^*$. 
\end{proof}

The second Lemma is concerned with uniqueness of solutions with the same boundary conditions at $r=R$. In particular, we want to show that if a function vanishes at $R$, its derivative there must be non-zero, unless the function is identically zero. The proof proceeds along the same lines as above, but is slightly simpler since it suffices to consider here the case where the potential  is fixed to be $U_{\minP}$, with $\minP$ the unique minimizer of the Pekar functional, i.e., we only consider the linearized equation.

\begin{lem}
\label{signofder}
The derivative of $\minP$ satisfies
\begin{equation}\label{limpp}
\lim_{r\nearrow R} \minP'(r)=c \hspace{3mm} \text{for some} \hspace{3mm} c<0.
\end{equation}
\end{lem} 

\begin{proof}
Integrating Eq.~\eqref{ELrad} using  that $\minP'(0)=0$, we have
\begin{equation}
\minP'(r)=\frac 1 {r^2} \int_0^r s^2\left(2 U_{\minP}(s) - \nu_{\minP}\right)\minP(s) ds.
\end{equation} 
From this we deduce that the limit in \eqref{limpp} exists and is finite, and by the monotonicity of $\minP$ it must be non-positive.  Suppose that $\minP'(R)=0$ and consider the function $\sigma(r):=r\minP(r)$, which then satisfies
\begin{equation}
\begin{cases}
&\sigma''=\left( 2 U_{\minP} - \nu_{\minP} \right)\sigma\hspace{10mm} \text{in}\hspace{3mm}[0,R]\\
&\sigma(R)=0,\\
&\sigma'(R)=0.
\end{cases}
\end{equation}
Using Taylor expansion (w.r.t. $R$) with remainder in integral form, we have
\begin{equation}
\sigma(r)=\int_r^R(s-r)\sigma''(s)ds=\int_r^R(s-r)\left(2 U_{\minP}(s) -\nu_{\minP} \right)\sigma(s)ds.
\end{equation}
Since $U_{\minP}$ is bounded, 
\begin{equation}
|\sigma(r)|\leq C(R-r)^2\|\sigma\|_{L^{\infty}([r,R])},
\end{equation}
which implies that $\sigma\equiv 0$ on $[\bar{r},R]$ if $\bar{r}$ is such that $C(R-\bar{r})^2<1$. This is a contradiction since  $\minP>0$ on $B_R$.
\end{proof}

\section{Convergence Results}\label{appB}

In this appendix we shall show that the Pekar minimizer $\phi_R$ and its energy $E_R$ converge to the corresponding full space quantities as $R\to \infty$. Recall that we have shown above that, for each $R>0$, there exists a unique positive minimizer $\phi_R$  of $\En$ (for $L^2$-normalized functions in $\Hball$).
On the other hand, it was shown in \cite{lieb1977existence} that there exists a unique positive and radial $\Psi$ minimizing the full space Pekar functional 
\begin{equation}
\label{PekarSpace}
\mathscr{E}(\phi)=\int_{\Rtre}|\nabla \phi|^2 dx -  \int_{\Rtre} \int_{\Rtre}\frac{|\phi(x)|^2|\phi(y)|^2}{|x-y|} dx dy=:T(\phi)-W(\phi).
\end{equation} 
 (for $L^2$-normalized functions in $\Hspace$). Our goal is to show that  $\phi_R \to \Psi$ (in $\Hspace$-norm, as well as pointwise) as $R\to \infty$, and that $E_R\to E_{\infty}:=\mathscr{E}(\Psi)$. We start with the latter.

\begin{prop}\label{approp}
$\lim_{R\to\infty} E_R = E_{\infty}$
\end{prop}

\begin{proof}
We start by approximating $\Psi$ with functions in $\Hball$. Consider the  sequence of cutoff functions
\begin{equation}
\eta_R(x)=
\begin{cases}
1 &x\in B_{R/2},\\
\frac{2(R-|x|)}{R} & x\in B_R\setminus B_{R/2},\\
0 &x\in B_R^c. 
\end{cases}
\end{equation}
We claim that $\Psi_R:=\eta_R\Psi \to \Psi$ in $H^1(\Rtre)$ and that
$\En(\Psi_R)\to\mathscr{E}(\Psi)=E_{\infty}$.
The $L^2$-convergence of $\Psi_R$ to $\Psi$ is immediate. Moreover, 
\begin{align}\nonumber
&\int_{\Rtre}|\nabla(\Psi_R-\Psi)|^2 dx\leq 2 \int_{\Rtre} |(\eta_R-1)\nabla \Psi|^2 dx+2  \int_{\Rtre} |\Psi\nabla \eta_R|^2 dx\\
&\leq 2\int_{B_{R/2}^c} |\nabla \Psi|^2 dx+\frac 8 {R^2} \int_{B_R\setminus B_{R/2}} |\Psi|^2dx\to 0
\end{align} 
as $R\to\infty$, showing the $H^1$-convergence.
To show $\En(\Psi_R)\to\mathscr{E}(\Psi)$, we first observe that $H^1$-convergence implies the convergence of the $L^2$-norms of the gradients and hence that $T_R(\Psi_R)\to T(\Psi)$. Moreover, from Newton's Theorem and the fact that the functions $\Psi_R$ are radial, we get
\begin{equation}
W_R(\Psi_R)=W(\Psi_R)+\frac 1 R\|\Psi_R\|_2^4.
\end{equation}
We can then apply dominated convergence to show $W(\Psi_R)\to W(\Psi)$ and conclude that our claim holds. 

It is now straightforward to conclude convergence of the minima. Indeed, we have
\begin{equation}
E_{\infty}\leftarrow\En(\Psi_R)=\|\Psi_R\|_2^2\left(T_R(\Psi_R/\|\Psi_R\|)-\|\Psi_R\|_2^2W_R(\Psi_R/|\Psi_R|)\right)\geq \|\Psi_R\|_2^2E_R. 
\end{equation} 
On the other hand, 
\begin{equation}
E_R=\En(\phi_R) \geq \mathscr{E}(\phi_R)\geq E_{\infty}.
\end{equation}
Therefore, necessarily $E_R\to E_{\infty}$.
\end{proof}

From the the previous Proposition, we readily deduce that $\phi_R$ is a minimizing sequence for the full space Pekar functional~\eqref{PekarSpace}. We can then proceed as in the proof of \cite[Theorem 7]{lieb1977existence} to conclude that $\phi_R$ is converging to $\Psi$ pointwise, weakly in $H^1(\Rtre)$ and strongly in $L^2(\Rtre)$. This latter statement implies also the convergence of the interaction energies, and since we have already proven the convergence of the full energies in Prop.~\ref{approp}, we conclude that also $\|\nabla \phi_R\|_2\to \|\nabla \Psi\|_2$. In combination with weak $H^1$-convergence, this implies strong $H^1$-convergence, and thus completes the proof of the convergence of the minimizers.

\begin{rem}
It is also possible to frame this discussion in the language of $\Gamma$-convergence of the functionals $\En$ to $\mathscr{E}$ w.r.t. the $H^1(\Rtre)$-norm. The corresponding liminf inequalities are readily shown to hold, and it is possible to recast the cutoff argument to construct recovery sequences for any $\Phi \in H^1(\Rtre)$ (which requires a little extra work for non-radial functions). In order to deduce the convergence of minimizers,  one still needs to employ the methods in \cite{lieb1977existence} in order to conclude equi-mild-coercivity of the functionals, however. 
\end{rem}

\bigskip
\noindent {\it Acknowledgments.} We are grateful for the hospitality at the Mittag-Leffler Institute, where part of this work has been done. 
Financial support through  the European Research Council (ERC) under the European Union's Horizon 2020 research and innovation programme (grant agreement No 694227) is acknowledged.

\end{document}